\documentclass[reqno]{amsart}

\makeatletter
\@namedef{subjclassname@2020}{%
  \textup{2020} Mathematics Subject Classification}
\makeatother

\usepackage{amsmath}
\usepackage{amssymb}
\usepackage{enumerate}

\newtheorem{Theorem}{Theorem}[section]

\newtheorem{Proposition}[Theorem]{Proposition}

\newcommand{\thref}[1]{Theorem \ref{#1}}

\newcommand{\reref}[1]{Remark \ref{#1}}
\newcommand{\seref}[1]{Section \ref{#1}}
\newcommand{\prref}[1]{Proposition \ref{#1}}

\theoremstyle{definition}
\newtheorem{Definition}[Theorem]{Definition}
\newtheorem{Remark}[Theorem]{Remark}
\newtheorem{Example}[Theorem]{Example}
\newtheorem*{Remark*}{Remark}

\numberwithin{equation}{section}

\begin{document}

\newcommand{\pd}{\partial}
\newcommand{\res}{\mathrm{res}}
\newcommand{\alg}{\mathrm{alg}}
\newcommand{\Gr}{\mathrm{Gr}}
\newcommand{\Grad}{\mathrm{Gr^{ad}}}
\newcommand{\Ai}{\mathrm{Ai}}
\newcommand{\Span}{\mathrm{span}}
\newcommand{\clspan}{\overline{\mathrm{span}}}
\newcommand{\Id}{\mathrm{Id}}
\newcommand{\diag}{\mathrm{diag}}
\newcommand{\Ad}{\mathrm{Ad}}
\newcommand{\type}{\mathrm{type}}

\newcommand{\Pset}{\mathbb{P}}
\newcommand{\Qset}{\mathbb{Q}}
\newcommand{\Rset}{\mathbb{R}}
\newcommand{\Tset}{\mathbb{T}}

\newcommand{\Cset}{\mathbb{C}}
\newcommand{\Nset}{\mathbb{N}}
\newcommand{\Zset}{\mathbb{Z}}

\newcommand{\tT}{\mathrm{T}}

\newcommand{\fA}{\mathfrak{A}}

\newcommand{\fE}{\mathfrak{E}}
\newcommand{\fF}{\mathfrak{F}}

\newcommand{\fG}{\mathfrak{G}}
\newcommand{\fH}{\mathfrak{H}}
\newcommand{\fHo}{\mathfrak{H}^{(1)}}
\newcommand{\fHt}{\mathfrak{H}^{(2)}}

\newcommand{\fFD}{\mathfrak{F}^{D}}
\newcommand{\fFK}{\mathfrak{F}^{K}}
\newcommand{\fK}{\mathcal{K}}
\newcommand{\fR}{\mathfrak{R}}
\newcommand{\fS}{\mathfrak{S}}

\newcommand{\fa}{\mathfrak{a}}
\newcommand{\fb}{\mathfrak{b}}
\newcommand{\fp}{\mathfrak{p}}
\newcommand{\fr}{\mathfrak{r}}

\newcommand{\fe}{\mathfrak{e}}
\newcommand{\fh}{\mathfrak{h}}
\newcommand{\ff}{\mathfrak{f}}
\newcommand{\fho}{\mathfrak{h}^{(1)}}
\newcommand{\fht}{\mathfrak{h}^{(2)}}

\newcommand{\cC}{\mathcal{C}}
\newcommand{\cE}{\mathcal{E}}
\newcommand{\cF}{\mathcal{F}}
\newcommand{\cG}{\mathcal{G}}
\newcommand{\cH}{\mathcal{H}}
\newcommand{\cK}{\mathcal{K}}
\newcommand{\cL}{\mathcal{L}}
\newcommand{\cM}{\mathcal{M}}
\newcommand{\cP}{\mathcal{P}}
\newcommand{\cQ}{\mathcal{Q}}
\newcommand{\cR}{\mathcal{R}}
\newcommand{\cS}{\mathcal{S}}
\newcommand{\cU}{\mathcal{U}}
\newcommand{\cV}{\mathcal{V}}
\newcommand{\cLD}{\mathcal{L}^{D}}
\newcommand{\cLK}{\mathcal{L}^{K}}

\newcommand{\cGD}{\mathcal{G}^{D}}
\newcommand{\cGK}{\mathcal{G}^{K}}

\newcommand{\sH}{\mathsf{H}}
\newcommand{\sK}{\mathsf{K}}
\newcommand{\sM}{\mathsf{M}}

\newcommand{\al}{\alpha}
\newcommand{\ka}{\kappa}
\newcommand{\la}{\lambda}
\newcommand{\om}{\omega}
\newcommand{\Ga}{\Gamma}

\newcommand{\Pt}{\tilde{P}}
\newcommand{\tS}{\tilde{S}}
\newcommand{\Vt}{\tilde{V}}
\newcommand{\pt}{\tilde{p}}
\newcommand{\xt}{\tilde{x}}
\newcommand{\et}{\tilde{e}}
\newcommand{\mt}{\tilde{m}}
\newcommand{\kt}{\tilde{k}}
\newcommand{\nt}{\tilde{n}}
\newcommand{\wt}{\tilde{w}}
\newcommand{\Gt}{\tilde{G}}
\newcommand{\Ht}{\tilde{H}}
\newcommand{\phit}{\tilde{\phi}}
\newcommand{\thetat}{\tilde{\theta}}

\newcommand{\Lt}{\tilde{L}}

\newcommand{\bka}{\boldsymbol{\kappa}}
\newcommand{\bla}{\boldsymbol{\lambda}}
\newcommand{\bmu}{\boldsymbol{\mu}}
\newcommand{\bnu}{\boldsymbol{\nu}}
\newcommand{\be}{\boldsymbol{e}}
\newcommand{\bp}{\boldsymbol{p}}

\newcommand{\bv}{\boldsymbol{v}}
\newcommand{\bw}{\boldsymbol{w}}
\newcommand{\bx}{\boldsymbol{x}}
\newcommand{\by}{\boldsymbol{y}}
\newcommand{\bz}{\boldsymbol{z}}
\newcommand{\bal}{\boldsymbol{\alpha}}
\newcommand{\bW}{\boldsymbol{W}}

\newcommand{\tte}{\mathtt{e}}

\newcommand{\rM}{\mathrm{M}}

\newcommand{\Vac}{v}

\newcommand{\cVdn}[3]{{\mathcal{V}}^{#1}_{#2}(#3)}

\newcommand{\tFt}[5]{{}_3F_2\left(\begin{matrix} #1 , #2, #3 \\
#4, #5 \end{matrix}\,; 1\right)}
\newcommand{\tFo}[4]{{}_2F_1\left(\begin{matrix} #1 , #2 \\
#3\end{matrix}\,; #4 \right)}

\title{Gaudin models for some classical multivariate distributions}

\author[P.~Iliev]{Plamen~Iliev}
\address{School of Mathematics, Georgia Institute of Technology, 
Atlanta, GA 30332--0160, USA}
\email{iliev@math.gatech.edu}
\thanks{The author gratefully acknowledges support from the Simons Foundation, and the hospitality and support of the Institut des Hautes \'Etudes Scientifiques (IHES) during a research visit in 2026.} 

\date{August 19, 2026}



\begin{abstract} 
Representations of the Kohno--Drinfeld Lie algebra associated with several classical multivariate distributions have played an important role in recent developments in the theory of quantum superintegrable systems.  In this work, we analyze the corresponding Gaudin models for the multivariate Hahn, Dirichlet, multinomial, and negative multinomial distributions. More precisely, using tools from representation theory, we construct multivariate orthogonal polynomials with respect to these distributions as common eigenfunctions of Gaudin operators. The polynomials are parametrized by solutions to the Bethe ansatz equations or, equivalently, by the roots of Heine--Stieltjes polynomials.
\end{abstract}

\maketitle

\section{Introduction} \label{se1}
Recently, there have been numerous developments exploring new relations between quantum integrable systems and multivariate extensions of classical orthogonal polynomials \cite{DGVV,DIV,I4,IX3,KMP2,MPW}. In particular, the analysis of the generic superintegrable system on the sphere in \cite{I3} revealed the role played by a representation of the Kohno--Drinfeld Lie algebra on the spaces of orthogonal polynomials associated with the Dirichlet distribution. 
These representations were extended to several discrete multivariate distributions in \cite{IX3}, which related earlier works \cite{GI,KM,Mi,Tr,Tr2} on multivariate analogs of the classical orthogonal polynomials and their spectral properties to maximal abelian subalgebras of the Kohno--Drinfeld Lie algebra generated by Jucys--Murphy elements. 

At the same time, there have been remarkable advances in the theory of Gaudin models, with far-reaching applications across different branches of mathematics and theoretical physics; see, for instance, \cite{FFR,Fr,MTV1,MTV2} and the references therein. A natural question is whether we can also link this theory to multivariate classical distributions by exploring the Kohno--Drinfeld structure mentioned above. The first step in this direction was the diagonalization of the Gaudin operators for the multinomial and the negative multinomial distributions in \cite{I5}. The approach developed in that work used an appropriate ansatz on the variety that parametrizes the multivariate Krawtchouk and Meixner polynomials defined in \cite{Gr1,Gr2,MT} and their spectral properties derived in \cite{I1,I2}, but an explicit connection to Gaudin's original work \cite{Gau1} remained elusive.

The goal of the present paper is to provide a direct link between the theory of Gaudin models and multivariate extensions of the classical orthogonal polynomials for classical distributions, which include the multivariate Hahn, Dirichlet, multinomial, and the negative multinomial distributions. In particular, we construct new families of multivariate Hahn and Jacobi polynomials as common eigenfunctions of commuting partial difference and differential operators, respectively, parametrized by solutions to the Bethe ansatz equations in Gaudin's work \cite{Gau1} or, equivalently, by the roots of the Heine–Stieltjes polynomials \cite{Heine,Stieltjes}. We also give a new Lie-algebraic framework and an alternative description of the multivariate Krawtchouk and Meixner polynomials introduced in \cite{I5}. 

The paper is organized as follows. In the next section, we briefly review the Kohno--Drinfeld Lie algebra and its connection to multivariate orthogonal polynomials and quantum integrability. In \seref{se3} we collect key properties of a Gaudin model associated with $\mathfrak{gl}_{2}$ needed for the multivariate Hahn distribution. We use the well-known ideas and techniques to build the Gaudin eigenvectors \cite{FFR,Fr,Gau1}, but we adapt the construction of the modules so that the statements suit our purposes. We also outline the classical Heine--Stieltjes problem and its connections to the simplicity of the spectrum of the Gaudin model needed in the following section. In \seref{se4}, we discuss the diagonalization of the Gaudin operators for the multivariate Hahn distribution and derive formulas for the corresponding polynomials. In \seref{se5}, we obtain families of multivariate Jacobi polynomials that diagonalize Gaudin operators for the Dirichlet distribution by an appropriate limit from the Hahn case. In \seref{se6}, we discuss the multinomial distribution in detail. In the first subsection, we briefly explain the underlying algebraic framework which simplifies in this case and the (complicated) Bethe ansatz equations reduce to the significantly simpler problem of finding the roots of a certain polynomial. Next, we illustrate how this approach leads to a new construction of the multivariate Krawtchouk polynomials defined in \cite{I5}, while also providing a direct proof of the Aomoto--Gelfand hypergeometric representation. The modifications needed for the negative multinomial distribution and the corresponding multivariate Meixner polynomials are discussed in \reref{re6.5}.

\subsection*{Notations}
Throughout the paper we will use the following notations:
\begin{itemize}
\item $\Rset$ and $\Cset$ will denote the fields of real and complex numbers, respectively. 
\item $\Nset_0$ will denote the set of all nonnegative integers.
\item For a vector $\bv=(v_1,\dots,v_k)$ we set $|\bv|=v_1+\cdots+v_k$.
\item For $a\in\Cset$ and  $k\in\Nset_0$ we will denote by $(a)_k $  the
Pochhammer symbol $(a)_k = a(a+1)\cdots(a+k-1)$.
\item For $\bx=(x_1,\dots,x_d)$ we will denote by $\Rset[x_1,\dots,x_d]$ or simply by $\Rset[\bx]$ the ring of polynomial in $x_1,\dots,x_d$ with real coefficients.
\item For $\bx=(x_1,\dots,x_d)$ and $k\in\Nset_0$ we will denote by $\Rset_k[\bx]$ the space of polynomial in $x_1,\dots,x_d$ with real coefficients of total degree at most $k$.
\item For $N\in\Nset $ we denote by $V_N^d$ the discrete simplex
\begin{equation*}
V_N^d= \{\bx \in \Nset_0^d: |\bx|=x_1+\cdots+x_{d} \le N\} \text{ in }\Rset^{d}.
\end{equation*}
\item $\cM_{n\times d}$ is the set of all $n\times d$ matrices $A=(a_{j,k})$ such that $a_{j,k}\in\{0,1\}$ and $A$ has at most one nonzero entry in each row. 
\item For $N\in\Nset $ we denote by ${\rM}_{d\times d}(N)$ the set of all $d\times d$ matrices $A=(a_{i,j})$ with nonnegative integer entries $a_{i,j}$, such that $\sum_{i,j=1}^{d}a_{i,j}\leq N$.
\end{itemize}

\section{Kohno--Drinfeld Lie algebra and multivariate orthogonal polynomials} \label{se2} 

\subsection{Kohno--Drinfeld Lie algebra}
Recall that the Kohno--Drinfeld Lie algebra \cite{Dr,Ko} is the Lie algebra $\mathfrak{D}_{d+1}$ with generators $\cL_{i,j}=\cL_{j,i}$, $i\neq j\in\{0,1,\dots,d\}$ and relations
\begin{subequations}\label{2.1}
\begin{align}
[\cL_{i,j},\cL_{k,l}]&=0, &&\text{ if }i,j,k,l \text{ are distinct,} \label{2.1a}\\
[\cL_{i,j},\cL_{i,k}+\cL_{j,k}]&=0, &&\text{ if }i,j,k \text{ are distinct.}   \label{2.1b}
\end{align}
\end{subequations}

Two maximal abelian subalgebras have been extensively studied in the literature which we briefly discuss below.

\subsubsection{Jucys--Murphy subalgebras} Using the equations in \eqref{2.1} it is easy to see that the {\em Jucys--Murphy elements}
\begin{equation}\label{2.2}
\cL_{0,1}, \cL_{0,2}+\cL_{1,2}, \cL_{0,3}+\cL_{1,3}+\cL_{2,3},\dots, \sum_{j=0}^{d-1}\cL_{j,d},
\end{equation}
are pairwise commuting and thus generate an abelian subalgebra.

\subsubsection{Gaudin subalgebras} Fix distinct complex numbers $\al_0,\dots,\al_{d}$ and let $\bal=(\al_0,\dots,\al_{d})$. For $i\in\{0,1,\dots,d\}$ we define {\em Gaudin elements} by
\begin{equation}\label{2.3}
\fG_i(\bal)=\sum_{\begin{subarray}{c}j=0\\ j\neq i \end{subarray}}^{d}\frac{\cL_{i,j}}{\al_{i}-\al_{j}}.
\end{equation}
It is straightforward to check that the elements $\fG_i(\bal)$ pairwise commute and thus generate an abelian subalgebra of $\mathfrak{D}_{d+1}$.

\begin{Remark}\label{re2.1}
From the Kohno--Drinfeld relations \eqref{2.1} it follows that the element 
\begin{equation}\label{2.4}
\cL=\sum_{0\leq i< j\leq d} \cL_{i,j}
\end{equation}
is central and belongs to all maximal abelian subalgebras of $\mathfrak{D}_{d+1}$. If we think of $\cL$ as a quantum Hamiltonian, then the operators $\cL_{i,j}$  represent integrals of motion for $\cL$. Clearly, we can get $\cL$ by adding all elements in \eqref{2.2} in the Jucys--Murphy case. For the Gaudin subalgebras we can obtain $\cL$ using the identity
\begin{equation}\label{2.5}
\cL=\sum_{i=1}^{d} (\al_i-\al_0) \fG_{i}(\bal).
\end{equation}
\end{Remark}

\subsection{Multivariate orthogonal polynomials}
Suppose that we have an inner product defined on the space $\Rset[\bx]=\Rset[x_1,\dots,x_d]$ or $\Rset_{N}[\bx]$ , where $\Rset_k[\bx]$ denotes the space of polynomial of total degree at most $k$. If $d>1$, there is no canonical way to define orthogonal polynomials, and therefore it is more natural to develop the general theory by working with the spaces ${\mathcal{V}}^{d}_{n}=\Rset_{n}[\bx]\ominus \Rset_{n-1}[\bx]$ of orthogonal polynomials of total degree $n$, where $n\in\Nset_0$ or $n\in\{0,\dots, N\}$, see \cite{DX}. For the multivariate Hahn, the multinomial and the negative  multinomial distributions, one can define partial difference operators $\cL_{i,j}$ which satisfy the Kohno--Drinfeld relations \eqref{2.1} and preserve the spaces ${\mathcal{V}}^{d}_{n}$, see \cite{IX3}. Furthermore, the spaces ${\mathcal{V}}_{n}^{d}$ can be characterized by a spectral equation of the form 
\begin{equation}
\cL(P(\bx))=\la_{n} P(\bx)\quad \text{ for every }\quad P(\bx)\in {\mathcal{V}}_{n}^{d},
\end{equation}
where $\cL$ is the operator in \eqref{2.4} and appropriate eigenvalues $\la_n$. This equation characterizes the distributions mentioned above and plays a central role in the classification of multivariate extensions of the discrete classical orthogonal polynomials \cite{IX1}. For the Dirichlet distribution similar properties hold for appropriate differential operators $\cL_{i,j}$ which can be obtained by a limit from the Hahn distribution, and the operator $\cL$ can be transformed into the Hamiltonian for the generic quantum superintegrable system on the sphere. In view of this, the operator $\cL$ for the multivariate Hahn distribution can be regarded as a discrete extension of the generic quantum superintegrable system on the sphere, while an appropriate limit from the operator $\cL$ for the multinomial distribution leads to the quantum harmonic oscillator \cite{IX3}.

\subsubsection{Diagonalization of the Jucys--Murphy elements} Multivariate orthogonal polynomials $P_{\bnu}(\bx)$ for the multivariate Hahn, Dirichlet and the multinomial distribution which are common eigenfunctions of the Jucys--Murphy elements in \eqref{2.2} 
have a long history, starting with the papers \cite{MP,Prorial} in dimension two for the Dirichlet distribution which contain formulas for bivariate Jacobi polynomials defined as a product of hypergeometric functions of a single variable. Similar formulas, providing orthogonal bases given explicitly as products of hypergeometric functions in arbitrary dimension, for each of the distributions above have appeared in numerous applications in probability and mathematical physics \cite{KMT,KM,Mi}. In \cite{GI} it was shown that all these polynomials are bispectral in the sense of \cite{DG}, i.e. they are common eigenfunctions of the Jucys--Murphy operators acting on the variables $\bx=(x_1,\dots,x_d)$, and at same time, they are common eigenfunctions of commuting partial difference operators acting on the degree indices $\bnu=(\nu_1,\dots,\nu_{d})$. This follows from the more general duality and bispectrality established in \cite{GI} for the multivariate Racah and Wilson polynomials defined by Tratnik \cite{Tr,Tr2}. 

\subsubsection{Diagonalization of the Gaudin subalgebras}
For the Gaudin algebras, the only known results so far concern the multinomial and the negative multinomial distributions in \cite{I5}. The approach developed in this work used the multivariate Krawtchouk and Meixner polynomials defined in \cite{Gr1,Gr2,MT}, their spectral properties derived in \cite{I1,I2}, and a suitable ansatz on the variety that parametrizes these polynomials, which in a broad sense is similar to the Bethe ansatz \cite{Bethe}.  These multivariate polynomials are also bispectral as shown in \cite{I5}, thus revealing an interesting duality for the corresponding Gaudin operators. In the last section, we explain how the polynomials defined in \cite{I5} can be obtained by a limit from the new families of Hahn polynomials constructed here, thus providing an explicit link to Gaudin's work \cite{Gau1}.
 The spectral properties of $P_{\bnu}(\bx)$ as functions of $\bnu$ for the Gaudin operators associated with the Hahn and Dirichlet operators amounts to a challenging new problem related to properties of the roots of the Heine--Stieltjes polynomials  \cite{Heine,Stieltjes}.

\section{Gaudin model for the Lie algebra $\mathfrak{gl}_{2}$} \label{se3}

\subsection{Algebraic setting} 
We fix a Lie algebra $\mathfrak{g}$ with basis $\{\fe,\ff,\fho,\fht\}$ and brackets
\begin{subequations}\label{3.1}
\begin{align}
&[\fho,\ff]=\ff, && [\fht,\ff]=\ff, && [\ff,\fe]=\fho+\fht, \label{3.1a}\\
&[\fho,\fe]=-\fe, && [\fht,\fe]=-\fe, && [\fho,\fht]=0. \label{3.1b}
\end{align}
\end{subequations}
In particular, $\fho-\fht$ belongs to the center of $\mathfrak{g}$. 
It is easy to see that $\mathfrak{g}$ is isomorphic to $\mathfrak{gl}_{2}$ if we map:
\begin{align*}
\fe\to \left[\begin{matrix} 0 &1 \\ 0 & 0\end{matrix}\right], \quad \ff\to \left[\begin{matrix} 0 &0 \\ 1 & 0\end{matrix}\right], 
\quad \fho\to \left[\begin{matrix} -1 &0 \\ 0 & 0\end{matrix}\right],  
\quad \fht\to \left[\begin{matrix} 0 &0 \\ 0 & 1\end{matrix}\right].
\end{align*}

Let $U(\mathfrak{g})$ be the universal enveloping algebra of $\mathfrak{g}$ and note that
\begin{equation}\label{3.2}
\ff \fe+\fho\fht-\fho \quad\text{ belongs to the center of } U(\mathfrak{g}).
\end{equation}

For $\la\in\Cset$, let $V_{\la}$ be a 
$\mathfrak{g}$-module generated by the vacuum vector 
$\Vac_{\la}$ 
\begin{equation}\label{3.3}
V_{\la}=U(\mathfrak{g})\cdot\Vac_{\la},
\end{equation}
and the relations
\begin{subequations}\label{3.4}
\begin{align}
&(\fht-\fho)\cdot \Vac_{\la}=\la\, \Vac_{\la}, \label{3.4a} \\
&(\fho-\ff)\cdot\Vac_{\la} =(\fht+\fe)\cdot\Vac_{\la} =0 \label{3.4b}.
\end{align}
\end{subequations}
For the applications in the next section, it will be convenient to work with the module $V_{\la}$ above instead of the finite-dimensional irreducible highest-weight modules parametrized by dominant integral weights for  $\mathfrak{sl}_{2}$ in the usual approach \cite{FFR,Fr}.

We fix a nonnegative integer $d$ and we set $\bla=(\la_0,\la_1,\dots,\la_d)$. We denote by $V_{\bla}$ the tensor product 
$V_{\la_0}\otimes V_{\la_1}\otimes\cdots \otimes V_{\la_d}$. 
For any element $a\in \mathfrak{g}$ and $j\in\{0,\dots,d\}$ we denote by $a_{j}$ the operator 
$$1\otimes \cdots \otimes \underbrace{a}_{j+1}\otimes \cdots \otimes1$$ 
which acts as $a$ on the $(j+1)$-st factor of $V_{\bla}$ and as the identity on all other factors.

Let $\al_0,\al_1,\dots,\al_d$ be distinct complex numbers and let $\bal=(\al_0,\dots,\al_d)$. The {\em Gaudin operators} are defined by 
\begin{equation}\label{3.5}
\fG_i(\bal)=\sum_{\begin{subarray}{c}j=0\\ j\neq i \end{subarray}}^{d}\frac{\fe_{i}\ff_{j}+\ff_{i}\fe_{j}+\fho_{i}\fht_{j} +\fht_{i} \fho_{j}}{\al_{i}-\al_{j}}, \qquad \text{for }\quad i=0,1,\dots,d.
\end{equation}
They pairwise commute and the main problem is to find the eigenvectors and the eigenvalues of these operators. From \eqref{3.4b} it follows that 
\begin{equation*}
\Vac_{\bla}=\Vac_{\la_0}\otimes \cdots \otimes \Vac_{\la_d}
\end{equation*}
is an eigenvector of the Gaudin operators with eigenvalue $0$. The main idea of the Bethe ansatz method \cite{Bethe,Gau1} is to produce new eigenvectors by applying elementary operators to the vacuum vector $\Vac_{\bla}$. To this end, it will be more convenient to work with generating functions \cite{FFR}. For a formal variable $u$ we set
\begin{equation}\label{3.6}
\cS(u)=\sum_{i=0}^{d}\frac{\fG_i(\bal)}{u-\al_i}=\sum_{0\leq i<j\leq d} \frac{\fe_{i}\ff_{j}+\ff_{i}\fe_{j}+\fho_{i}\fht_{j} +\fht_{i} \fho_{j}}{(u-\al_i)(u-\al_j)},
\end{equation}
and
\begin{subequations}\label{3.7}
\begin{align}
& \fF(u) =\sum_{j=0}^{d} \frac{\ff_j}{u-\al_j},&&  \fE(u) =\sum_{j=0}^{d} \frac{\fe_j}{u-\al_j}, \label{3.7a}\\
& \fHo(u) =\sum_{j=0}^{d} \frac{\fho_j}{u-\al_j},&&  \fHt(u) =\sum_{j=0}^{d} \frac{\fht_j}{u-\al_j}. \label{3.7b}
\end{align}
\end{subequations}

The relations \eqref{3.1} can be rewritten in terms of the generating functions as follows

\begin{align*}
& [\fHo(v),\fF(u)] =-\frac{1}{u-v} (\fF(u)-\fF(v)), \quad  [\fHo(v),\fE(u)] =\frac{1}{u-v} (\fE(u)-\fE(v)), \\
& [\fHt(u),\fF(v)] =-\frac{1}{u-v} (\fF(u)-\fF(v)), \quad [\fHt(u),\fE(v)] =\frac{1}{u-v} (\fE(u)-\fE(v)), \\ 
&  [\fHo(u),\fHt(v)] =0, \quad  [\fF(u),\fE(v)] =-\frac{1}{u-v} (\fHo(u)+\fHt(u)-\fHo(v)-\fHt(v)),\\ 
\end{align*}
and 
\begin{equation*}
\cS(u)=\fF(u)\fE(u)+\fHo(u)\fHt(u)+\frac{\pd \fHo(u)}{\pd u}-\sum_{j=0}^{d}\frac{\ff_j\fe_j+\fho_j\fht_j-\fho_j}{(u-\al_j)^2}.
\end{equation*}

From the above equations and \eqref{3.2} it follows that

\begin{equation}\label{3.8}
[\cS(u),\fF(v)]=\frac{1}{u-v} (\fF(v)(\fHo(u)+\fHt(u))-\fF(u)(\fHo(v)+\fHt(v))).
\end{equation}
If we set 
$$\fK(u)=\fHo(u)-\fF(u),\text{ and }\fH(u)=\fHt(u)-\fHo(u),$$
then we can rewrite \eqref{3.8} as 
\begin{equation}\label{3.9}
[\cS(u),\fF(v)]=\frac{1}{u-v} (2\fF(v)\fK(u)-2\fF(u)\fK(v)-\fF(u)\fH(v)+\fF(v)\fH(u)).
\end{equation}
Note that 
\begin{equation}\label{3.10}
[\fH(u),\fF(v)]=0,\qquad \text{ and }\qquad \fH(u)\cdot \Vac_{\bla}=\left(\sum_{j=0}^{d}\frac{\la_j}{u-\al_j}\right) \Vac_{\bla}.
\end{equation}
Since 
\begin{equation}\label{3.11}
\fK(u)\cdot \Vac_{\bla}=0, \qquad \text{ and }\qquad [\fK(u),\fF(v)]=-\frac{1}{u-v} (\fF(u)-\fF(v)),
\end{equation}
one can show by induction that 
\begin{align}
\fK(u)\cdot(\fF(w_n) \cdots \fF(w_{1})\cdot \Vac_{\bla})=\sum_{j=1}^{n}\frac{1}{w_j-u} 
\cdot \fF(w_{n})\cdots  \underbrace{\fF(u)}_{j} \cdots \fF(w_{1})\cdot \Vac_{\bla}& \nonumber\\
+\left(\sum_{j=1}^{n}\frac{1}{u-w_j}\right) \fF(w_n) \cdots \fF(w_{1})\cdot \Vac_{\bla}.&\label{3.12}
\end{align}
Set $\bw=(w_1,\dots,w_n)$ and let
\begin{equation}\label{3.13}
r(u; \bw)=\sum_{j=1}^{n}\frac{1}{u-w_j} 
\qquad \text{ and }\qquad 
\fr(u;\bla,\bal)=\sum_{k=0}^{d}\frac{\la_k}{u-\al_k}.
\end{equation}

Using equations \eqref{3.9}-\eqref{3.12} it follows by induction on $n$ that 
\begin{subequations}\label{3.14}
\begin{align}
&\cS(u)\cdot(\fF(w_n) \cdots \fF(w_{1})\cdot \Vac_{\bla}) = \rho(u;\bw,\bla,\bal) \,\fF(w_n) \cdots \fF(w_{1})\cdot \Vac_{\bla} \nonumber\\
&\qquad\qquad +\sum_{j=1}^{n}\frac{f_j^{(n)}(\bw;\bla;\bal)}{w_j-u}   \fF(w_{n})\cdots  \underbrace{\fF(u)}_{j} \cdots \fF(w_{1})\cdot \Vac_{\bla}, \label{3.14a}
\end{align}
where 
\begin{equation}\label{3.14b}
\rho(u;\bw,\bla,\bal)=r(u; \bw) \fr(u;\bla,\bal) + r^2(u; \bw)+ \frac{d r}{d u }(u;\bw),
\end{equation}
and 
\begin{equation}\label{3.14c}
f_j^{(n)}(\bw;\bla;\bal)=\sum_{\begin{subarray}{c}k=1\\ k\neq j \end{subarray}}^{n}\frac{2}{w_j-w_k}+\fr(w_j;\bla,\bal).
\end{equation}
\end{subequations}
Clearly, if the coefficients $f_j^{(n)}(\bw;\bla;\bal)$ in \eqref{3.14c} all vanish, then $\fF(w_n)\cdot \cdots \fF(w_{1})\cdot \Vac_{\bla} $ is an eigenvector of $\cS(u)$ with eigenvalue $\rho(u;\bw,\bla,\bal)$ given in \eqref{3.14b}. We summarize these results in the proposition below.

\begin{Proposition}\label{pr3.1}
If $w_1,\dots,w_n$ satisfy the Bethe ansatz equations 
\begin{equation}\label{3.15}
\sum_{\begin{subarray}{c}k=1\\ k\neq j \end{subarray}}^{n}\frac{2}{w_j-w_k}+\sum_{\ell=0}^{d}\frac{\la_{\ell}}{w_j-\al_{\ell}}=0 ,\qquad \text{ for }\quad j=1,\dots,n,
\end{equation}
then 
\begin{equation}\label{3.16}
\Vac_{\bla}(\bw)=\fF(w_n)\cdots  \fF(w_{1})\cdot \Vac_{\bla}
\end{equation}
is a common eigenvector of the Gaudin operators \eqref{3.5}. 
Moreover,
\begin{equation}\label{3.17}
\cS(u)\cdot  \Vac_{\bla}(\bw) = \rho(u;\bw,\bla,\bal) \, \Vac_{\bla}(\bw),
\end{equation}
where $\rho(u;\bw,\bla,\bal)$ is given in \eqref{3.14b}.
\end{Proposition}
Equivalently, computing the residue of both sides of \eqref{3.17} at $u=\al_i$ we can rewrite equation \eqref{3.17} as 
\begin{equation}\label{3.18}
\fG_i(\bal)\cdot  \Vac_{\bla}(\bw) = (\la_i r(\al_i; \bw))\, \Vac_{\bla}(\bw), \qquad \text{for }i=0,\dots,d.
\end{equation}

\subsection{Real parameters and solutions, and simplicity of the joint spectrum of Gaudin operators} 
For the application in the next section, we are interested in parameters satisfying the following conditions
\begin{equation}\label{3.19}
\la_j>0 \quad \text{ and }\quad \al_j\in\Rset \quad\text{ for }\quad  j=0,\dots,d.
\end{equation}
In this case, one can show that the solutions of the Bethe ansatz equations \eqref{3.15} are real, and the joint spectrum of Gaudin operators is simple, see \cite{SV} for the $\mathfrak{sl}_2$ case, \cite{MTV2} for far reaching extensions, and \cite{Var} for equations extending \eqref{3.15} corresponding to critical points for product of powers of linear functions. Below, we outline a proof of this fact with the notations above, following \cite{Szego} and providing a link to the classical Heine--Stieltjes problem \cite{Heine, Stieltjes}. 

Since $[\fF(w_i),\fF(w_j)]=0$, the eigenvectors $\Vac_{\bla}(\bw)$ in \eqref{3.16} are invariant under permutations of the parameters $w_1,w_2,\dots,w_n$ and thus we can try to characterize the polynomials 
\begin{equation}\label{3.20}
p(w)=p(w;\bw)=(w-w_1)(w-w_2)\cdots (w-w_{n}),
\end{equation}
whose roots satisfy equations \eqref{3.15}. Note that 
$$\sum_{\begin{subarray}{c}k=1\\ k\neq j \end{subarray}}^{n}\frac{2}{w_j-w_k}=\frac{p''(w_j)}{p'(w_j)},$$
hence if we set 
\begin{equation}\label{3.21}
A(w)=\prod_{j=0}^{d}(w-\al_j), \qquad \text{ and }\qquad B(w)=A(w)\left(\sum_{k=0}^{d}\frac{\la_k}{w-\al_k}\right),
\end{equation}
then $A(w)$ and $B(w)$ are polynomials of degrees $\deg _wA(w)=d+1$, $\deg _wB(w)=d$, and \eqref{3.15} can be rewritten as 
$$A(w_j)p''(w_j)+B(w_j)p'(w_j)=0,\qquad \text{ for }\quad j=1,\dots,n.$$
The last equation is equivalent to the condition that $A(w)p''(w)+B(w)p'(w)$ is divisible by $p(w)$. Thus, \eqref{3.15} holds if and only if there exists a polynomial $C(w)$, of degree $\deg _wC(w)\leq d-1$ such that 
\begin{equation}\label{3.22}
A(w)p''(w)+B(w)p'(w)-C(w)p(w)=0,
\end{equation}
which is the problem studied by Heine~\cite{Heine}, see also \cite[page~151]{Szego}. Equation \eqref{3.15} is known as the Stieltjes electrostatic interpretation.

If we fix $A(w)$ and $B(w)$ as in \eqref{3.21} with parameters satisfying \eqref{3.19} and if $C(w)$ is a polynomial, one can show that
\begin{enumerate}[(1)]
\item Equation~\eqref{3.22} can have at most one monic polynomial solution $p(w)$, and\label{fact1}
\item If $\sigma\in S_{d+1}$ is a permutation such that $\al_{\sigma(0)}<\al_{\sigma(1)}<\cdots <\al_{\sigma(d)}$, then the roots $\{w_j\}$ of a polynomial $p(w)$ satisfying \eqref{3.22} are real and are  distributed in the open intervals $(\al_{\sigma(0)}, \al_{\sigma(1)})$, $(\al_{\sigma(1)}, \al_{\sigma(2)})$, \dots, $(\al_{\sigma(d-1)}, \al_{\sigma(d)})$,
\end{enumerate}
see \cite[pages 152--153]{Szego}. 
Furthermore, one can show that if $\bnu=(\nu_1,\dots,\nu_d)$ is a weak composition of $n$ into $d$ parts, i.e.
if $\nu_1,\dots,\nu_d$ are nonnegative integers such that 
\begin{equation*}
\nu_1+\cdots+\nu_d=n,
\end{equation*}
then there exists a unique monic polynomial solution $p(w)=p_{\bnu}(w)$ of \eqref{3.22} having exactly $\nu_k$ roots in the $k$-th interval $(\al_{\sigma(k-1)}, \al_{\sigma(k)})$, for every $k=1,\dots,d$,
see \cite[pages 153--155]{Szego}. 

Finally, note that if we consider the partial fraction decomposition of $C(w)/A(w)$:
\begin{equation*}
\frac{C(w)}{A(w)}=\sum_{k=0}^{d}\frac{\mu_k}{w-\al_k},
\end{equation*}
and if we use \eqref{3.21}, we can rewrite \eqref{3.22} as follows
\begin{equation}\label{3.23}
p''(w)+\left(\sum_{k=0}^{d}\frac{\la_k}{w-\al_k}\right)p'(w)-\left(\sum_{k=0}^{d}\frac{\mu_k}{w-\al_k}\right)p(w)=0.
\end{equation}
If $p_{\bnu}(w)$ is the unique polynomial solution having exactly $\nu_k$ roots in the $k$-th interval $(\al_{\sigma(k-1)}, \al_{\sigma(k)})$, then the residue at $w=\al_i$ of the equation above gives
$$\mu_i=\la_i\frac{p'(\al_i)}{p(\al_i)}=\la_i r(\al_i; \bw),$$
using the notation in \eqref{3.13}.
This combined with \eqref{3.18} shows that $\mu_i=\mu_i(\bnu)$ is the eigenvalue of the Gaudin operator $\fG_i(\bal)$. Thus, the differential equation \eqref{3.22} is completely determined by the eigenvalues of the Gaudin operators, and therefore the uniqueness of the polynomial solution stated in (\ref{fact1}) above proves the simplicity of the joint spectrum.  Dividing \eqref{3.23} by $w^{n-1}$ and computing the limit $w\to\infty$ we see that $\mu_0+\mu_1+\cdots+\mu_d=0$ (which also follows from the fact that $\sum_{i=0}^{d}\fG_i(\bal)=0$). In view of this, we will work with $(\mu_1,\dots,\mu_d)$ and we set 
$$\bmu(\bnu)=(\mu_1,\dots,\mu_d).$$
If we replace $\mu_0$ by $-(\mu_1+\cdots+\mu_d)$ in \eqref{3.23}, we can rewrite the differential equation for $p(w)$ as
\begin{equation}\label{3.24}
p''(w)+\left(\sum_{k=0}^{d}\frac{\la_k}{w-\al_k}\right)p'(w)-\left(\sum_{k=1}^{d}\frac{\mu_k (\al_k-\al_0)}{(w-\al_k)(w-\al_0)}\right)p(w)=0.
\end{equation}
Dividing the last equation by $w^{n-2}$ and computing the limit $w\to\infty$ we find
\begin{equation}\label{3.25}
\sum_{k=1}^{d}(\al_k-\al_0)\mu_k=n\left(n+|\bla|-1\right).
\end{equation}

We summarize these statements in the theorem below.
\begin{Theorem}\label{th3.2}
Suppose that $\{\la_j\}_{j=0}^{d}$ are positive numbers and $\{\al_j\}_{j=0}^{d}$ are distinct real parameters. Then for every weak composition $\bnu=(\nu_1,\dots,\nu_d)$ of $n$ into $d$ parts there exists a unique monic polynomial $p_{\bnu}(w)$ of degree $n$, whose zeros $\{w_j\}$ satisfy the Bethe ansatz equations \eqref{3.15} and exactly $\nu_k$ of $\{w_1,\dots,w_n\}$ belong to $k$-th interval $(\al_{\sigma(k-1)}, \al_{\sigma(k)})$, for every $k=1,\dots,d$. The eigenvalues $\mu(\bnu)=(\mu_1,\dots,\mu_d)$ of Gaudin operators $\fG_1(\bal),\dots,\fG_d(\bal)$ corresponding to the eigenvector \eqref{3.16} built from the roots of $p_{\bnu}(w)$ satisfy \eqref{3.25}. 
If $\bnu=(\nu_1,\dots,\nu_d)$ and $\bnu'=(\nu_1',\dots,\nu_d')$ are different weak compositions of $n$ into $d$ parts, 
then $\bmu(\bnu)\neq\bmu(\bnu')$.
 \end{Theorem}

\section{Multivariate Hahn distribution} \label{se4}
\subsection{Inner product and spaces of orthogonal polynomials}
Suppose that $\ka_0,\ka_1,\dots,\ka_d$ are real numbers such that $\ka_j>-1$ for $j=0,\dots,d$. 
The probability mass function of the multivariate Hahn distribution with parameters $\bka=(\ka_0,\ka_1,\dots,\ka_d)$ and $N\in\Nset$ is
$$\frac{N!}{(|\bka|+ d+1)_N} \prod_{j=0}^{d} \frac{(\ka_j+1)_{x_j}}{x_j!}, \text{ where }x_i\in\Nset_0 \text{ and }x_0+\cdots+x_{d}=N.$$
We set 
\begin{equation}\label{4.1}
x_{0}=N-(x_1+\cdots+x_{d}), 
\end{equation}
and we will work with the independent variables $\bx=(x_1,\dots,x_d)$. This leads to the multivariate Hahn weight
\begin{subequations}\label{4.2}
\begin{equation}\label{4.2a}
\sH_{\bka,N} (\bx) = \frac{N!}{(|\bka|+ d+1)_N}  \frac{(\ka_0+1)_{N-|\bx|}}{(N-|\bx|)!}  \prod_{j=1}^{d} \frac{(\ka_j+1)_{x_j}}{x_j!}, 
\end{equation}
defined on the discrete simplex $V_N^d=\{\bx \in \Nset_0^d: |\bx|=x_1+\cdots+x_{d} \le N\}$.
With the above notation, we consider the space $\Rset_N[\bx]$ of polynomials of total degree at most $N$ equipped with the inner product
\begin{equation}\label{4.2b}
\langle f,g\rangle_{\sH_{\bka,N} }=\sum_{\bx\in V_N^d}f(\bx)g(\bx)\sH_{\bka,N} (\bx),
\end{equation}
and for $n\in\{0,1,\dots,N\}$ we denote by $\cVdn{d}{n}{\sH_{\bka,N}}$ the corresponding space of orthogonal polynomials of degree $n$ defined by 
\end{subequations}
\begin{equation*}
\cVdn{d}{n}{\sH_{\bka,N}}=\Rset_{n}[\bx]\ominus\Rset_{n-1}[\bx],
\end{equation*}
with the convention that $\Rset_{-1}[\bx]=\{0\}$.
Note that 
\begin{equation}\label{4.3}
\dim (\cVdn{d}{n}{\sH_{\bka,N}})=\binom{n+d-1}{n}.
\end{equation}

\subsection{Kohno--Drinfeld structure}\label{ss4.2}
We denote by $\{\be_1,\be_2,\dots,\be_d\}$ the standard basis for $\Rset^d$, and by $E_{x_i}$ and $E_{x_i}^{-1}$ the shift operators 
acting on a function $f(\bx)$ as follows
\begin{align*}
E_{x_i}f(\bx)=f(\bx+\be_i) \quad \text{and}\quad E_{x_i}^{-1}f(\bx)=f(\bx-\be_i).
\end{align*}
Following \cite{IX3}, for $i\neq j\in\{0,\dots,d\}$ we define
\begin{equation}\label{4.4}
\cL_{i,j} =  (x_i+\ka_i+1) x_j (E_{x_i}E_{x_j}^{-1}-\Id)+ (x_j+\ka_j+1) x_i (E_{x_j}E_{x_i}^{-1}-\Id),
\end{equation}
with the convention that $E_{x_{0}}=\Id$ and $x_{0}=N-|\bx|$. One can show that these operators satisfy the Kohno--Drinfeld relations \eqref{2.1}. Moreover, 
\begin{enumerate}
\item these operators are self-adjoint with respect to the inner product \eqref{4.2}, and 
\item they preserve the space $\Rset_n[\bx]$ of polynomials of $x_1,\dots,x_d$ of total degree $n$, i.e. $\cL_{i,j}:\Rset_n[\bx]\to\Rset_n[\bx]$ for every $n\in\{0,\dots,N\}$ and  $i\neq j\in\{0,\dots,d\}$.
\end{enumerate}
Thus, for every $n\in \{0,\dots,N\}$, the operators in \eqref{4.4} define a representation of the Kohno--Drinfeld Lie algebra $\mathfrak{D}_{d+1}$ on the space of orthogonal polynomials $\cVdn{d}{n}{\sH_{\bka,N}}$. 

\begin{Remark}\label{re4.1}
For the representation above, the operator $\cL$ in \eqref{2.4} can be written as follows
\begin{align}
&\cL=\sum_{0\leq i< j\leq d} \cL_{i,j}=\sum_{1\leq i\neq j\leq d}(x_i+\ka_i+1)x_j(E_{x_i}E_{x_j}^{-1}-\Id)\nonumber \\
&\quad+\sum_{i=1}^d(x_i+\ka_i+1)(N-|\bx|)(E_{x_i}-\Id) + \sum_{i=1}^d(N-|\bx|+\ka_{0}+1)x_i(E_{x_i}^{-1}-\Id).\label{4.5}
\end{align}
One can show that for $n=0,1,\dots,N$ we have
\begin{equation}\label{4.6}
\cL(P(\bx))=-n(n+|\bka|+d)P(\bx)\quad \text{ for every }\quad P(\bx)\in \cVdn{d}{n}{\sH_{\bka,N}},
\end{equation}
which characterizes the multivariate Hahn distribution \cite[Section 5]{IX1}.
\end{Remark}

\begin{Remark}\label{re4.2}
By an appropriate limit and gauge transformation, we can transform the operator $\cL$ in \eqref{4.5} into the Hamiltonian for the generic quantum superintegrable system on the sphere, see \reref{re5.1}. Thus $\cL$ can be regarded as a discrete quantum Hamiltonian which extends the generic quantum superintegrable system on the sphere, and $\cL_{i,j}$ represent integrals of motion, or symmetries for $\cL$, see \cite{IX3} for algebraic relations and explicit generators of the symmetry algebra.
\end{Remark}

\subsection{Gaudin algebras}\label{ss4.3}
We fix distinct real numbers $\al_0,\dots,\al_{d}$ and our goal is to find a basis of $\cVdn{d}{n}{\sH_{\bka,N}}$ which diagonalizes the Gaudin operators
\begin{equation}\label{4.7}
\cG_i(\bal)=\sum_{\begin{subarray}{c}j=0\\ j\neq i \end{subarray}}^{d}\frac{\cL_{i,j}}{\al_{i}-\al_{j}}, \qquad i=0,\dots, d.
\end{equation}
In particular, limits and reductions of the operator for $i=0$
\begin{equation}\label{4.8}
\cG_0(\bal)=\sum_{j=1}^{d}\frac{1}{\al_{0}-\al_{j}}\left((x_j+\ka_j+1) (N-|\bx|) (E_{x_j}-\Id)+(N-|\bx|+\ka_0+1) x_j (E_{x_j}^{-1}-\Id)\right), 
\end{equation}
have been studied in the context of birth and death processes, see \reref{re6.6}.

For $i\in \{0,\dots, d\}$, the operators 
\begin{align*}
&\ff_{i}=x_{i}E_{x_i}^{-1}, && \fe_{i}=-(x_{i}+\ka_{i}+1)E_{x_i},\\
&\fho_{i}=x_{i}, && \fht_{i}=x_{i}+\ka_{i}+1,
\end{align*}
define a representation of the Lie algebra $\mathfrak{gl}_2$ with the basis in \eqref{3.1} on the space of polynomials $\Rset[x_i]$ satisfying \eqref{3.4} with $\la_{i}=\ka_{i}+1$ and vacuum vector $\Vac_{\la_i}=1$. Note that with these notations, the operator in \eqref{4.4} can be rewritten as follows
\begin{equation}\label{4.9}
\cL_{i,j} =-(\fe_{i}\ff_{j}+\ff_{i}\fe_{j}+\fho_{i}\fht_{j} +\fht_{i}\fho_{j} ),
\end{equation}
and therefore, up to an overall sign, the Gaudin operators in \eqref{4.7} coincide with the ones defined by \eqref{3.5}. 
Thus, if we consider $x_{0},x_{1},\dots,x_{d}$ as independent variables, we can try to diagonalize the Gaudin operators \eqref{4.7} using \prref{pr3.1} on the space $\Rset[x_{0},x_{1},\dots,x_{d}]$. For the representation above, the operator $\fF(w)$ in \eqref{3.7a} takes the form 
$$\fF(w) =\sum_{j=0}^{d} \frac{1}{w-\al_j}\,x_{j}E_{x_j}^{-1},$$
and we can construct common eigenvectors for the Gaudin operators of the form 
$$\fF(w_k) \cdots \fF(w_{1})\cdot 1,$$
treating $x_{0},x_{1},\dots,x_{d}$ as independent variables, and eliminating $x_0$ at the end, using \eqref{4.1} in the resulting polynomial in $\Rset[x_{0},x_{1},\dots,x_{d}]$, after all operators have been applied. If we want to work within the space $\Rset[x_{1},\dots,x_{d}]$ we can ``restore" the commutativity between the functions and the operators depending on or acting on $x_0$ and the ones depending or acting on $x_{i}$ for $i\in \{1,\dots, d\}$ by treating the polynomials also as functions of $N$, and by replacing $\fF(w)$ above with 
\begin{equation}\label{4.10}
\fF(w) = \frac{1}{w-\al_0}\,(N-|\bx|)E_{N}^{-1}+\sum_{j=1}^{d} \frac{1}{w-\al_j}\,x_{j}E_{x_j}^{-1}E_{N}^{-1}, 
\end{equation}
where $|\bx|=x_{1}+\cdots+x_{d}\in \Rset[x_{1},\dots,x_{d}]$ and  $E_N^{-1}$ denotes the shift operator acting on functions depending on $N$ by $E_N^{-1} f(N) = f(N-1)$. Applying \prref{pr3.1} with the operator $\fF(w)$ in \eqref{4.10}, it follows that if 
$w_1,\dots,w_n$ satisfy the Bethe ansatz equations 
\begin{equation}\label{4.11}
\sum_{\begin{subarray}{c}s=1\\ s\neq j \end{subarray}}^{n}\frac{2}{w_j-w_s}+\sum_{\ell=0}^{d}\frac{\ka_{\ell}+1}{w_j-\al_{\ell}}=0 ,\qquad \text{ for }\quad j=1,\dots,n,
\end{equation}
then the polynomial $\fF(w_n) \cdots \fF(w_{1})\cdot 1$ is a common eigenvector of the Gaudin operators \eqref{4.7}. 
Since the parameters are real and $\la_j=\ka_j+1>0$, we know that every solution $\bw=(w_1,\dots,w_n)$ of \eqref{4.11} has real components, and for every weak composition $\bnu=(\nu_1,\dots,\nu_d)$ of $n$ into $d$ parts we have exactly one solution $\bw(\bnu)=(w_1,\dots,w_n)$, modulo permutations, which has $\nu_k$ components in the $k$-th interval $(\al_{\sigma(k-1)}, \al_{\sigma(k)})$, for every $k=1,\dots,d$. 

\begin{Definition}
For every weak composition $\bnu=(\nu_1,\dots,\nu_d)$ of $n$ into $d$ parts we define
\begin{equation}\label{4.12}
P_{\bnu}(\bx;\bka,N;\bal)=\frac{\prod_{j=1}^{n}(w_{j}^{\bnu}-\al_{0})}{(-N)_{n}}\,\fF(w_n^{\bnu}) \cdots \fF(w_{1}^{\bnu})\cdot 1 \in\Rset[\bx],
\end{equation}
where $\bw(\bnu)=(w_1^{\bnu},\dots,w_n^{\bnu})$ is a solution of the Bethe ansatz equations \eqref{4.11} which has $\nu_k$ components in the $k$-th interval $(\al_{\sigma(k-1)}, \al_{\sigma(k)})$, for every $k=1,\dots,d$, and $\fF$ is the operator defined in \eqref{4.10}.
\end{Definition}
The polynomial $P_{\bnu}(\bx;\bka,N;\bal)$ above is normalized so that it has value $1$ at $\bx=0$. 

Recall that $\cM_{n\times d}$ denotes the set of all $n\times d$ matrices $A=(a_{j,k})$ such that $a_{j,k}\in\{0,1\}$ and $A$ has at most one nonzero entry in each row.  
With the above notations, we can formulate the main result in this section as follows.
\begin{Theorem}\label{th4.4}
If $\bw(\bnu)=(w_1^{\bnu},\dots,w_n^{\bnu})$ is a solution of the Bethe ansatz equations \eqref{4.11} which has $\nu_k$ components in the $k$-th interval $(\al_{\sigma(k-1)}, \al_{\sigma(k)})$, for every $k=1,\dots,d$, and if we set 
\begin{equation}
v_{j,k}^{\bnu}=\frac{\al_k-\al_0}{w_j^{\bnu}-\al_k} \quad \text{ for }\quad j=1,\dots,n,\quad k=1,\dots,d,
\end{equation}
then the polynomial in \eqref{4.12} can be written as
\begin{equation}\label{4.14}
P_{\bnu}(\bx;\bka,N;\bal)=\sum_{A=(a_{j,k})\in\cM_{n\times d}}\frac{\prod_{k=1}^{d}(-x_k)_{\sum_{j=1}^{n}a_{j,k}}}{(-N)_{\sum_{j=1}^{n}\sum_{k=1}^{d}a_{j,k}}}\, \prod_{j=1}^{n}\prod_{k=1}^d(v_{j,k}^{\bnu})^{a_{j,k}}.
\end{equation}
Moreover, the polynomials $\{P_{\bnu}(\bx;\bka,N;\bal):\bnu\in\Nset_0^{d}\text{ and }|\bnu|=n\}$ form an orthogonal basis of the space $\cVdn{d}{n}{\sH_{\bka,N}}$ and satisfy the spectral equations
\begin{equation}\label{4.15}
\cG_i(\bal) P_{\bnu}(\bx;\bka,N;\bal)= \mu_i(\bnu) P_{\bnu}(\bx;\bka,N;\bal), \quad\text{ for }\quad i=0,\dots,d
\end{equation}
where $\cG_i(\bal) $ are the Gaudin operators in \eqref{4.7}, and 
\begin{equation}\label{4.16}
\mu_i(\bnu) = -(\ka_i+1)\left(\sum_{j=1}^{n}\frac{1}{\al_i-w_j^{\bnu}}\right).
\end{equation}
 \end{Theorem}

\begin{proof}
First, we prove equation~\eqref{4.14} from \eqref{4.12} and \eqref{4.10} by induction on $n$. The fact that the entries of $\bw$ solve the Bethe ansatz equations (and their distribution depends on $\bnu$) is irrelevant for this part, so we will drop $\bnu$ to simplify the notation. 
Thus, assuming that \eqref{4.14} holds for some $\bnu$ such that $|\bnu|=n$, we want to show that the polynomial
\begin{align*}
P&=\frac{w_{n+1}-\al_{0}}{N}\, \fF(w_{n+1}) \cdot  P_{\bnu}(\bx;\bka,N;\bal)\\
&=\frac{w_{n+1}-\al_{0}}{N}\, \fF(w_{n+1}) \cdot \sum_{A=(a_{j,k})\in\cM_{n\times d}}\frac{\prod_{k=1}^{d}(-x_k)_{\sum_{j=1}^{n}a_{j,k}}}{(-N)_{\sum_{j=1}^{n}\sum_{k=1}^{d}a_{j,k}}}\, \prod_{j=1}^{n}\prod_{k=1}^dv_{j,k}^{a_{j,k}}
\end{align*}
can be written in the form \eqref{4.14} with $n$ replaced by $n+1$, i.e. 
\begin{equation}\label{4.17}
P=\sum_{B=(b_{j,k})\in\cM_{(n+1)\times d}}\frac{\prod_{k=1}^{d}(-x_k)_{\sum_{j=1}^{n+1}b_{j,k}}}{(-N)_{\sum_{j=1}^{n+1}\sum_{k=1}^{d}b_{j,k}}}\, \prod_{j=1}^{n+1}\prod_{k=1}^dv_{j,k}^{b_{j,k}}.
\end{equation}
Note that 
\begin{equation}\label{}
\frac{w_{n+1}-\al_{0}}{N}\, \fF(w_{n+1}) =\fF_1+ \fF_2(w_{n+1}),
\end{equation}
where
\begin{equation*}
\fF_1 = E_{N}^{-1} -\frac{1}{N}\sum_{s=1}^{d} x_s (\Id-E_{x_s}^{-1}) E_{N}^{-1} \quad\text{ and }\quad \fF_2(w_{n+1})=\frac{1}{N} \sum_{s=1}^{d}v_{n+1,s} x_{s} E_{x_s}^{-1} E_{N}^{-1}.
\end{equation*}
Using the identity $a(\Id-E_{a}^{-1})\cdot(-a)_{k}=k(-a)_k$ it follows that  
\begin{equation*}
\fF_1 P_{\bnu}(\bx;\bka,N;\bal) = P_{\bnu}(\bx;\bka,N;\bal)= \sum_{A=(a_{j,k})\in\cM_{n\times d}}\frac{\prod_{k=1}^{d}(-x_k)_{\sum_{j=1}^{n}a_{j,k}}}{(-N)_{\sum_{j=1}^{n}\sum_{k=1}^{d}a_{j,k}}}\, \prod_{j=1}^{n}\prod_{k=1}^dv_{j,k}^{a_{j,k}}
\end{equation*}
which gives all terms in \eqref{4.17} for which $B$ has zeros in the last row (and we can identify this subset of $\cM_{(n+1)\times d}$ with $\cM_{n\times d}$ by deleting the last row). Using the identity $aE_{a}^{-1}(-a)_{k}=-(-a)_{k+1}$ we see that 
 \begin{align*}
\fF_2(w_{n+1}) P_{\bnu}(\bx;\bka,N;\bal) &=\sum_{s=1}^{d}\sum_{A=(a_{j,k})\in\cM_{n\times d}}\frac{(-x_s)_{\sum_{j=1}^{n}a_{j,s}+1}\prod_{k\neq s}^{d}(-x_k)_{\sum_{j=1}^{n}a_{j,k}}}{(-N)_{\sum_{j=1}^{n}\sum_{k=1}^{d}a_{j,k}}+1}\\
&\qquad \times\, v_{n+1,s}\, \prod_{j=1}^{n}\prod_{k=1}^dv_{j,k}^{a_{j,k}}.
\end{align*}
The right-hand side of the last equation gives all terms in \eqref{4.17} for which $B$ has a nonzero entry in the last row. These matrices $B$ must have the form $B=A+E_{n+1,s}$ for some $s=1,\dots,d$, where $A$ has zeros in the last row (so we can identify $A$ with a matrix in $\cM_{n\times d}$) and $E_{n+1,s}$ is the $(n+1)\times d$ matrix whose $(n+1,s)$ entry is equal to $1$ and all other entries are zero. It is easy to see that if we rewrite the sums over $s$ and $A$ in the last equation as a sum over matrices $B=A+E_{n+1,s}$ in $\cM_{(n+1)\times d}$ with a nonzero entry in the last row we obtain the terms needed in \eqref{4.17}. This completes the proof of \eqref{4.14}. 

Using \eqref{2.5}, \thref{th3.2}, the identity \eqref{3.25} and \eqref{4.9}, we see that the polynomials $P_{\bnu}(\bx;\bka,N;\bal)$ satisfy \eqref{4.6} and therefore they belong to $\cVdn{d}{n}{\sH_{\bka,N}}$. 
Equation~\eqref{4.15} follows from \prref{3.1} and \eqref{3.18}. 
By \thref{th3.2}, if $\bnu\neq\bnu'$ there exists $i\in\{1,\dots,d\}$ such that $\mu_i(\bnu) \neq \mu_i(\bnu') $. Since the Gaudin operator $\fG_i(\bal)$ is self-adjoint with respect to the Hahn inner product \eqref{4.2}, it follows that the polynomials $P_{\bnu}(\bx;\bka,N;\bal)$ are mutually orthogonal, hence linearly independent. From \thref{th3.2} and \eqref{4.3} it follows that the number of polynomials is equal to the dimension of the space  $\cVdn{d}{n}{\sH_{\bka,N}}$ and therefore they form an orthogonal basis of $\cVdn{d}{n}{\sH_{\bka,N}}$.
\end{proof}

\begin{Example}[$n=2$]\label{ex4.5} 
If $E_{i,j}$ denotes the $2\times d$ matrix whose $(i,j)$ entry is equal to $1$ and all 
other entries are zero, then for $n=2$ we have 
$$\cM_{2\times d}=\{0\}\cup\{E_{i,k}:i=1,2, k=1,\dots,d \} \cup\{E_{1,k}+E_{2,j}:  j, k \in\{1,\dots,d\} \}. $$
Therefore,
\begin{align*}
P_{\bnu}(\bx;\bka,N;\bal)&=1+\sum_{k=1}^{d}\frac{v_{1,k}^{\bnu}+v_{2,k}^{\bnu}}{N}x_k\\
&\qquad+\sum_{1\leq j<k\leq d}\frac{v_{1,k}^{\bnu}v_{2,j}^{\bnu}+v_{1,j}^{\bnu}v_{2,k}^{\bnu}}{N(N-1)}x_jx_k
+\sum_{k=1}^{d}\frac{v_{1,k}^{\bnu}v_{2,k}^{\bnu}}{N(N-1)}x_k(x_k-1).
\end{align*}
\end{Example}

\begin{Remark}[Classical formula in dimension one]\label{re4.6}
When $d=1$ we have $\bx=x_1$ and $\bnu=n$. For fixed $n$, we omit the $n$ dependence to simplify the notation, and formula \eqref{4.14} reads
\begin{align}
P_{n}(x;\bka,N;\bal)&=1+\sum_{k=1}^{n}v_{k,1} \frac{x}{N}+\sum_{1\leq j<k\leq n }v_{j,1}v_{k,1} \frac{x(x-1)}{N(N-1)}\nonumber\\
&\qquad+\cdots + \left(\prod_{k=1}^{n}v_{k,1}\right) \frac{(-x)_n}{(-N)_n}\nonumber\\
&=1+\sum_{j=1}^{n}\tte_j(\bv) \frac{(-x)_j}{(-N)_j},\label{4.19}
\end{align}
where $\tte_j(\bv) $ are the elementary symmetric polynomials of $v_{1,1},v_{2,1},\dots,v_{n,1}$. Note that when $d=1$ we have $v_{j,1}=\frac{\al_1-\al_0}{w_j-\al_1}$ where $w_j$ are the roots of the polynomial solution of \eqref{3.24}. We have one independent  Gaudin operator $\fG_{1}(\bal)=-\fG_{0}(\bal)=\frac{\cL_{0,1}}{\al_0-\al_1}$ and we can normalize it by setting $\al_0=1$ and $\al_1=0$. Using \eqref{3.24}, \eqref{3.25} and \eqref{4.9} we see that $v_{j,1}=-1/w_j$ where $w_j$ are the roots of the polynomial solution of the Euler's hypergeometric differential equation 
$$w(w-1)p''(w)+((\ka_1+1)(w-1)+(\ka_0+1) w)p'(w)-n(n+\ka_0+\ka_1+1)p(w)=0,$$
which has $3$ regular singular points at $0$, $1$ and $\infty$. This simply means that $p_{n}(w)$ is the Jacobi polynomial of degree $n$ for which we have an explicit formula in terms of the Gauss' hypergeometric function:
\begin{align*}
p_{n}(w)&=\frac{(-1)^n(\ka_1+1)_n}{(n+\ka_0+\ka_1+1)_n}\,\tFo{-n}{n+\ka_0+\ka_1+1}{\ka_1+1}{w}\\
&=\frac{(-1)^n(\ka_1+1)_n}{(n+\ka_0+\ka_1+1)_n}\,\sum_{j=0}^{n} \frac{(-n)_j(n+\ka_0+\ka_1+1)_j}{j!(\ka_1+1)_j}\,w^j.
\end{align*} 
From the formula for $p_{n}(w)$ above and Vieta's formulas it follows that 
$$\tte_j(\bv) =\frac{(-n)_j(n+\ka_0+\ka_1+1)_j}{j!(\ka_1+1)_j}.$$
Substituting the last equation into \eqref{4.19} we find the usual hypergeometric representation for the Hahn polynomial
\begin{align*}
P_{n}(x;\bka,N;\bal)&= 1+\sum_{j=1}^{n}\frac{(-n)_j(n+\ka_0+\ka_1+1)_j(-x)_j}{j!(-N)_j(\ka_1+1)_j}\\
&=\tFt{-n}{n+\ka_0+\ka_1+1}{-x}{-N}{\ka_1+1}.
\end{align*}
\end{Remark}

\section{Dirichlet distribution}\label{se5}
We fix parameters $\bka=(\ka_0,\dots,\ka_{d})$ such that $\ka_j>-1$ for all $j\in\{0,\dots,d\}$. For $\by=(y_1,\dots,y_d)$, the Dirichlet distribution on the simplex 
$$\tT^{d}=\{\by\in\Rset^{d}:  y_i\geq 0\text{ and }|\by|\leq 1\}$$
is given by the density
\begin{equation}\label{5.1}
W_{\bka} (\by)= \frac{\Ga(|\bka|+d+1)}{\prod_{j=0}^{d}\Ga(\ka_j+1)}\, y_1^{\ka_1} \cdots y_d^{\ka_d} (1-|\by|)^{\ka_{0}}.
\end{equation}
On the space $\Rset[\by]$ of polynomials of $y_1,y_2,\dots,y_d$, define an inner product by
\begin{equation}\label{5.2}
\langle f, g\rangle_{W_{\bka}} =\int_{\tT^d}f(\by)g(\by) W_{\bka} (\by)\, d\by,
\end{equation}
and for $n\in\Nset_0$ let 
\begin{equation}\label{5.3}
\cVdn{d}{n}{W_{\bka}}=\Rset_{n}[\by]\ominus\Rset_{n-1}[\by],
\end{equation}
denote the space of orthogonal polynomials of degree $n$. For $i\neq j\in\{0,\dots,d\}$ consider the partial differential operators
\begin{equation}\label{5.4}
\cLD_{i,j} =  y_iy_j(\pd_{y_i}-\pd_{y_j})^2+\left((\ka_i+1)y_j-(\ka_j+1)y_i\right)(\pd_{y_i}-\pd_{y_j}),
\end{equation}
where $\pd_{y_i}=\frac{\pd}{\pd y_i}$, for $i=1,\dots, d$, $\pd_{y_{0}}=0$ and $y_{0}=1-|\by|$. It is straightforward to show that these operators satisfy the Kohno--Drinfeld relations \eqref{2.1} and preserve the spaces $\cVdn{d}{n}{W_{\bka}}$ of orthogonal polynomials with respect to the Dirichlet distribution, i.e.  $\cLD_{i,j}:\cVdn{d}{n}{W_{\bka}}\to\cVdn{d}{n}{W_{\bka}}$ for every $n\in\Nset_0$ and  $i\neq j\in\{0,\dots,d\}$. In fact, many algebraic properties of these operators can be deduced from the Hahn operators \eqref{4.4} via a limiting process as follows. If we set 
\begin{equation}\label{5.5}
x_i=Ny_i,\qquad \text{ for }\qquad i=1,\dots, d,
\end{equation}
then one can show that the operators in \eqref{4.4} and \eqref{5.4} are related by
\begin{equation*}
\lim_{N\to\infty}\cL_{i,j}=\cLD_{i,j},
\end{equation*}
see \cite[Section~5.3]{GI}. This implies that the corresponding Gaudin operators for the Dirichlet distribution can also be obtained by a limit from the ones in \eqref{4.7}:
\begin{equation}\label{5.6}
\cGD_i(\bal)=\sum_{\begin{subarray}{c}j=0\\ j\neq i \end{subarray}}^{d}\frac{\cLD_{i,j}}{\al_{i}-\al_{j}}=\lim_{N\to\infty} \cG_i(\bal), \qquad i=0,\dots, d.
\end{equation}
In particular, the limit of the operator in \eqref{4.8} takes the form 
\begin{equation}\label{5.7}
\cGD_0(\bal)=\sum_{j=1}^{d}\frac{1}{\al_{0}-\al_{j}}\left(y_j(1-|\by|)\pd_{y_j}^{2}+\left((\ka_{j}+1)(1-|\by|)-(\ka_{0}+1)y_{j}\right)\pd_{y_j} \right).
\end{equation}

\begin{Remark}\label{re5.1}
For the representation of the Kohno--Drinfeld Lie algebra in \eqref{5.4}, the operator in \eqref{2.4} can be written as follows
\begin{align}
&\cLD=\sum_{0\leq i< j\leq d} \cLD_{i,j}=\sum_{i=1}^{d}y_i(1-y_i)\pd_{y_i}^2-2\sum_{1\leq i<j\leq d}y_iy_j\pd_{y_i}\pd_{y_j}\nonumber\\
&\qquad +\sum_{i=1}^{d}\left(\ka_i+1-(|\bka|+d+1)y_i\right)\pd_{y_i}.\label{5.8}
\end{align}
By an appropriate gauge transformation, ignoring inessential constant terms and factors, we can transform the operator $\cLD$ above into the Hamiltonian for the generic quantum superintegrable system on the sphere
\begin{equation}
\cLD\to \mathcal{H}=\Delta_{\mathbb{S}^d}+\sum_{i=0}^{d}\frac{\frac{1}{4}-\ka_{i}^2}{z_i^2},
\end{equation}
where $y_i=z_i^2$ and $\Delta_{\mathbb{S}^d}$ is the Laplace--Beltrami operator on the sphere. This system has been extensively studied in the literature \cite{I3,KMP2,KMT,MPW} as an important example of a second-order superintegrable system, possessing $(2d-1)$ second-order algebraically independent symmetries. The spaces $\cVdn{d}{n}{W_{\bka}}$ provide irreducible representations of the (symmetry) algebra generated by the first integrals \cite{I4}.
\end{Remark}

In order to write explicit formulas for the common eigenfunctions of the Gaudin operators, we set
\begin{equation}\label{5.10}
\fFD(w) = 1-|\by|+(w-\al_0)\,\sum_{j=1}^{d} \frac{y_j}{w-\al_j}, 
\end{equation}
which is a limit of the operator $\frac{w-\al_0}{N}\fF(w)$ defined in \eqref{4.10} as $N\to\infty$.  
\begin{Definition}
For every weak composition $\bnu=(\nu_1,\dots,\nu_d)$ of $n$ into $d$ parts we define
\begin{equation}\label{5.11}
P^{D}_{\bnu}(\by;\bka;\bal)=\fFD(w_n^{\bnu}) \cdots \fFD(w_{1}^{\bnu}) \in\Rset[\by],
\end{equation}
where $\bw(\bnu)=(w_1^{\bnu},\dots,w_n^{\bnu})$ is a solution of the Bethe ansatz equations \eqref{4.11} which has $\nu_k$ components in the $k$-th interval $(\al_{\sigma(k-1)}, \al_{\sigma(k)})$, for every $k=1,\dots,d$, and $\fFD$ is the polynomial defined in \eqref{5.10}.
\end{Definition}
Applying the change of variables \eqref{5.5} and taking the limit $N\to\infty$ in \thref{th4.4} we obtain the following result for the Gaudin model associated with the Dirichlet distribution.
\begin{Theorem}\label{th5.3}
If $\bw(\bnu)=(w_1^{\bnu},\dots,w_n^{\bnu})$ is a solution of the Bethe ansatz equations \eqref{4.11} which has $\nu_k$ components in the $k$-th interval $(\al_{\sigma(k-1)}, \al_{\sigma(k)})$, for every $k=1,\dots,d$, and if we set 
$$v_{j,k}^{\bnu}=\frac{\al_k-\al_0}{w_j^{\bnu}-\al_k} \quad \text{ for }\quad j=1,\dots,n,\quad k=1,\dots,d,$$ 
then the polynomial in \eqref{5.11} can be written as
\begin{equation}\label{5.12}
P^{D}_{\bnu}(\by;\bka;\bal)=\sum_{A=(a_{j,k})\in\cM_{n\times d}} \left(\prod_{j=1}^{n}\ \prod_{k=1}^d(v_{j,k}^{\bnu})^{a_{j,k}}\right)\,  \prod_{k=1}^{d}y_k^{\sum_{j=1}^{n}a_{j,k}}  .
\end{equation}
Moreover, the polynomials $\{P^{D}_{\bnu}(\by;\bka;\bal):\bnu\in\Nset_0^{d}\text{ and }|\bnu|=n\}$ form an orthogonal basis of the space $\cVdn{d}{n}{W_{\bka}}$ and satisfy the spectral equations
\begin{equation*}
\cGD_i(\bal) P^{D}_{\bnu}(\by;\bka;\bal)= \mu_i(\bnu) P^{D}_{\bnu}(\by;\bka;\bal), \quad\text{ for }\quad i=0,\dots,d
\end{equation*}
where $\cGD_i(\bal) $ are the Gaudin operators in \eqref{5.6}, and 
\begin{equation*}
\mu_i(\bnu) = -(\ka_i+1)\left(\sum_{j=1}^{n}\frac{1}{\al_i-w_j^{\bnu}}\right).
\end{equation*}
 \end{Theorem}

\section{Multinomial distribution}\label{se6}

For the Gaudin model associated with the multinomial distribution, the Bethe ansatz equations decouple and the constructions simplify significantly. We start by describing briefly the modifications needed in \seref{se3}, and then we explain the construction of the corresponding orthogonal polynomials. This provides an alternative approach to some of the results in \cite{I5}.

\subsection{Algebraic framework}\label{ss6.1}
We fix a Lie algebra $\mathfrak{g}$ with basis $\{\fe,\ff,\fho,\fht\}$ and brackets
\begin{subequations}\label{6.1}
\begin{align}
&[\fho,\ff]=\ff, && [\fht,\ff]=0, && [\ff,\fe]=\fht, \label{6.1a}\\
&[\fho,\fe]=-\fe, && [\fht,\fe]=0, && [\fho,\fht]=0. \label{6.1b}
\end{align}
\end{subequations}
In particular, $\fht$ belongs to the center of $\mathfrak{g}$ and $\ff \fe+\fho\fht$ belongs to the center of the universal enveloping algebra $U(\mathfrak{g})$.  
For $\la\in\Cset$, let $V_{\la}$ be a $\mathfrak{g}$-module generated by the vacuum vector 
$\Vac_{\la}$ 
\begin{equation}\label{6.2}
V_{\la}=U(\mathfrak{g})\cdot\Vac_{\la},
\end{equation}
and the relations
\begin{subequations}\label{6.3}
\begin{align}
&\fht\cdot \Vac_{\la}=\la\, \Vac_{\la}, \label{6.3a} \\
&(\fho-\ff)\cdot\Vac_{\la} =(\fht+\fe)\cdot\Vac_{\la} =0. \label{6.3b}
\end{align}
\end{subequations}

We use the same notation as in \seref{se3}: we fix a nonnegative integer $d$, set $\bla=(\la_0,\la_1,\dots,\la_d)$, and we denote by $V_{\bla}$ the tensor product 
$V_{\la_0}\otimes V_{\la_1}\otimes\cdots \otimes V_{\la_d}$. The Gaudin operators are defined by
\begin{equation}\label{6.4}
\fG_i(\bal)=\sum_{\begin{subarray}{c}j=0\\ j\neq i \end{subarray}}^{d}\frac{\fe_{i}\ff_{j}+\ff_{i}\fe_{j}+\fho_{i}\fht_{j} +\fht_{i} \fho_{j}}{\al_{i}-\al_{j}}, \qquad \text{for }\quad i=0,1,\dots,d,
\end{equation}
and 
\begin{equation*}
\cS(u)=\sum_{i=0}^{d}\frac{\fG_i(\bal)}{u-\al_i}=\sum_{0\leq i<j\leq d} \frac{\fe_{i}\ff_{j}+\ff_{i}\fe_{j}+\fho_{i}\fht_{j} +\fht_{i} \fho_{j}}{(u-\al_i)(u-\al_j)}.
\end{equation*}
Working with generating functions for the Lie algebra $\mathfrak{g}$ with relations in \eqref{6.1} it is easy to check that \eqref{3.9} is replaced by 
\begin{equation}\label{6.5}
[\cS(u),\fF(v)]=\frac{1}{u-v} (\fF(v)\fHt(u)-\fF(u)\fHt(v)).
\end{equation}
Setting $\bw=(w_1,\dots,w_n)$ and using the functions $r(u; \bw)$ and $\fr(u;\bla,\bal)$ in \eqref{3.13} one can show that \eqref{3.14a} in \seref{se3} is replaced by 

\begin{align}
&\cS(u)\cdot(\fF(w_n) \cdots \fF(w_{1})\cdot \Vac_{\bla}) = r(u; \bw) \fr(u;\bla,\bal) \,\fF(w_n) \cdots \fF(w_{1})\cdot \Vac_{\bla} \nonumber\\
&\qquad\qquad +\sum_{j=1}^{n}\frac{\fr(w_j;\bla,\bal)}{w_j-u}   \fF(w_{n})\cdots  \underbrace{\fF(u)}_{j} \cdots \fF(w_{1})\cdot \Vac_{\bla}. \label{6.6}
\end{align}
Note that the Bethe ansatz equations 
\begin{equation}\label{6.7}
\fr(w_j;\bla,\bal)=\sum_{k=0}^{d}\frac{\la_k}{w_j-\al_k}=0 ,\qquad \text{ for }\quad j=1,\dots,n,
\end{equation}
are independent of $n$, and reduce simply to the problem of finding the roots of the polynomial 
\begin{equation}\label{6.8}
p(w)=\prod_{k=0}^{d}(w-\al_k)\left(\sum_{k=0}^{d}\frac{\la_k}{w-\al_k}\right).
\end{equation}
Thus, solving the equation $p(w)=0$ (just once!) we obtain formulas for all Gaudin eigenvectors. The simplicity of the joint spectrum can also be deduced directly, without using the Heine--Stieltjes theory. 
The proposition below summarizes the above computations and represents an analog of \prref{pr3.1} and \thref{th3.2}.

\begin{Proposition}\label{pr6.1}
Suppose that $\la_{0},\dots,\la_{d}$ are positive real numbers and $\al_0,\dots,\al_d$ are distinct real numbers.
\begin{enumerate}[(i)]
\item Then the polynomial $p(w)$ in  \eqref{6.8} has $d$ distinct real roots $w_{1},\dots,w_{d}$. 
\item For $\bnu=(\nu_1,\dots,\nu_d)\in\Nset_0^d$ let
\begin{equation}\label{6.9}
\Vac_{\bla}(\bnu)=\fF(w_1)^{\nu_1}\cdots  \fF(w_{d})^{\nu_{d}}\cdot \Vac_{\bla}.
\end{equation}
Then $\Vac_{\bla}(\bnu)$ is  a common eigenvector of the Gaudin operators \eqref{6.4} and satisfies the spectral equations 
\begin{equation}\label{6.10}
\fG_i(\bal) \cdot  \Vac_{\bla}(\bnu) = \mu_i (\bnu) \Vac_{\bla}(\bnu), \quad\text{ for }\quad i=0,\dots,d
\end{equation}
where 
\begin{equation}\label{6.11}
\mu_i (\bnu)= \la_i \left(\sum_{j=1}^{d}\frac{\nu_{j}}{\al_{i}-w_{j}}\right).
\end{equation}
\item If $\bnu\neq\bnu'$ there exists $i\in\{1,\dots, d\}$ such that $\mu_i (\bnu)\neq \mu_i (\bnu')$.
\end{enumerate}
\end{Proposition}

\begin{proof}
If we arrange $\al_{j}$ in increasing order, i.e. if $\sigma\in S_{d+1}$ is a permutation such that $\al_{\sigma(0)}<\al_{\sigma(1)}<\cdots <\al_{\sigma(d)}$ then it is easy to see that on every subinterval $(\al_{\sigma(k-1)},\al_{\sigma(k)})$ the polynomial $p(w)$ has opposite signs at the end points, and therefore it will vanish inside the interval. This shows that $p(w)$ has exactly one real root in every subinterval $(\al_{\sigma(k-1)},\al_{\sigma(k)})$ for $k=1,\dots,d$ proving (i).
Equation \eqref{6.10} follows from \eqref{6.6} by computing the residue at $u=\al_{i}$. 
Finally, since $\la_{i}$ are  nonzero, if we assume that  $\mu_i (\bnu)= \mu_i (\bnu')$ for some $\bnu\neq\bnu'$ and all $i =1,\dots,d$,  it will follow that the rational function 
$R(w)=\sum_{j=1}^{d}\frac{\nu_{j}-\nu_{j}'}{w-w_j}$
has at most $d$ poles at the points $w_j$ and at least $d+2$ zeros at $\al_i$ and $\infty$, leading to a contradiction.
\end{proof}

\subsection{Multinomial inner product, polynomials and operators}\label{ss6.2}
We fix $N\in\Nset$ and parameters $\bp=(p_0,\dots,p_{d})$ such that $p_j>0$ for $j\in\{0,\dots,d\}$ and $|\bp|=p_0+\cdots+p_d=1$. 
The probability mass function of the multinomial distribution is given by
\begin{equation}\label{6.12}
\sK_{\bp,N} (\bx)  = \binom{N}{x_0,x_1,\dots,x_d} p_0^{x_0}p_1^{x_1}\cdots p_{d}^{x_{d}} =\frac{N!}{x_0!x_1!\cdots x_{d}!}\, p_0^{x_0}p_1^{x_1}\cdots p_{d}^{x_{d}},
\end{equation}
where $\bx=(x_1,\dots,x_d)\in V_N^d$ and $x_0=N-|\bx|$. The corresponding inner product on $\Rset_{N}[\bx]$ is defined by 
\begin{equation}\label{6.13}
\langle f,g\rangle_{\sK_{\bp,N} }=\sum_{\bx\in V_N^d}f(\bx)g(\bx)\sK_{\bp,N} (\bx),
\end{equation}
and for $n\in\{0,1,\dots,N\}$ we denote by $\cVdn{d}{n}{\sK_{\bp,N}}$ the corresponding space of orthogonal polynomials of degree $n$ 
\begin{equation}\label{6.14}
\cVdn{d}{n}{\sK_{\bp,N}}=\Rset_{n}[\bx]\ominus\Rset_{n-1}[\bx].
\end{equation}
Using the shift operators $E_{x_i}^{\pm 1}$ in Subsection~\ref{ss4.2}, for $i\neq j\in\{0,\dots,d\}$ we define the operators
\begin{equation}\label{6.15}
\cLK_{i,j} =  p_i x_j (E_{x_i}E_{x_j}^{-1}-\Id)+ p_j x_i (E_{x_j}E_{x_i}^{-1}-\Id),
\end{equation}
with the convention that $E_{x_{0}}=\Id$ and $x_{0}=N-|\bx|$. For every $n=0,1,\dots,N$, the operators in \eqref{6.15} define a representation of the Kohno--Drinfeld Lie algebra $\mathfrak{D}_{d+1}$  on the space $\cVdn{d}{n}{\sK_{\bp,N}}$. We can connect these operators to the Lie algebra in the previous subsection as follows. 
For $i\in \{0,\dots, d\}$, the operators 
\begin{align*}
&\ff_{i}=x_{i}E_{x_i}^{-1}, && \fe_{i}=-p_{i}E_{x_i},\\
&\fho_{i}=x_{i}, && \fht_{i}=p_{i},
\end{align*}
satisfy \eqref{6.1} and therefore define a representation of the Lie algebra $\mathfrak{g}$ on the space of polynomials $\Rset[x_i]$, satisfying \eqref{6.3} with $\la_{i}=p_{i}>0$ and vacuum vector $\Vac_{\la_i}=1$. With these notations, the operators in \eqref{6.15} can be written as
\begin{equation}\label{6.16}
\cLK_{i,j} =-(\fe_{i}\ff_{j}+\ff_{i}\fe_{j}+\fho_{i}\fht_{j} +\fht_{i}\fho_{j} ),
\end{equation}
and therefore for distinct real numbers $\al_0,\dots,\al_{d}$ the Gaudin operators for the multinomial distribution 
\begin{equation}\label{6.17}
\cGK_i(\bal)=\sum_{\begin{subarray}{c}j=0\\ j\neq i \end{subarray}}^{d}\frac{\cLK_{i,j}}{\al_{i}-\al_{j}}, \qquad i=0,\dots, d.
\end{equation}
differ by an overall sign from the ones in \eqref{6.4}. 
Using \prref{pr6.1} and the roots of the polynomial $p(w)$ in \eqref{6.8}, we can write an explicit orthogonal basis of the entire space $\Rset_{N}[\bx]$ consisting of eigenfunctions of the Gaudin operators in terms Aomoto--Gelfand hypergeometric series. Recall that ${\rM}_{d\times d}(N)$ denotes the set of all $d\times d$ matrices $A=(a_{i,j})$ with nonnegative integer entries $a_{i,j}$, such that $\sum_{i,j=1}^{d}a_{i,j}\leq N$.

\begin{Theorem}\label{th6.2}
Let  $w_1,\dots,w_d$ denote the roots of the equation
\begin{equation}
\sum_{k=0}^{d}\frac{p_k}{w-\al_k}=0,
\end{equation}
which are simple and real, and let
$$\omega_{i,j}=\frac{\al_0-\al_i}{w_j-\al_i} \quad \text{ for }\quad i,j \in\{1,\dots,d\}.$$ 
For $\bnu\in V_N^d$, define
\begin{equation}\label{6.19}
P^{K}_{\bnu}(\bx;\bp,N;\bal)=\sum_{A=(a_{i,j})\in {\rM}_{d\times d}(N)}\frac{\prod_{j=1}^{d}(-\nu_j)_{\sum_{i=1}^{d}a_{i,j}}\, \prod_{i=1}^{d}(-x_i)_{\sum_{j=1}^{d}a_{i,j}}}{(-N)_{\sum_{i,j=1}^{d}a_{i,j}}}
\prod_{i,j=1}^{d}\frac{\omega_{i,j}^{a_{i,j}}}{a_{i,j}!}.
\end{equation}
The polynomials $\{P^{K}_{\bnu}(\bx;\bp,N;\bal):\bnu\in V_N^d\}$ form an orthogonal basis of the space $\Rset_{N}[\bx]$ with respect to the inner product \eqref{6.13} and satisfy the spectral equations
\begin{equation}\label{6.20}
\cGK_i(\bal) P^{K}_{\bnu}(\bx;\bp,N;\bal)= \mu_i(\bnu) P^{K}_{\bnu}(\bx;\bp,N;\bal), \quad\text{ for }\quad i=0,\dots,d,
\end{equation}
where $\cGK_i(\bal) $ are the Gaudin operators in \eqref{6.17}, and 
\begin{equation}\label{6.21}
\mu_i(\bnu) = - p_i \left(\sum_{j=1}^{d}\frac{\nu_{j}}{\al_{i}-w_{j}}\right).
\end{equation}
 \end{Theorem}

\begin{proof}
Similarly to the arguments in \seref{se4} we see that, if we define 
\begin{equation}\label{}
\fFK(w) = \frac{1}{N}\,(N-|\bx|)E_{N}^{-1}+\frac{w-\al_0}{N}\sum_{j=1}^{d} \frac{1}{w-\al_j}\,x_{j}E_{x_j}^{-1}E_{N}^{-1},
\end{equation}
then by \prref{pr6.1} the polynomials 
\begin{equation}\label{6.23}
P^{K}_{\bnu}(\bx;\bp,N;\bal)= \fFK(w_1)^{\nu_1}\cdots  \fFK(w_{d})^{\nu_{d}}\cdot 1
\end{equation}
will satisfy \eqref{6.20} and therefore they form an orthogonal basis of the space $\Rset_{N}[\bx]$. The remaining nontrivial part is to establish the Aomoto--Gelfand hypergeometric series representation given in \eqref{6.19}. This can be deduced easily from Theorem 4.1 in \cite{I5} using the fact the spectral equations \eqref{6.20} combined with the normalization $P^{K}_{\bnu}(0;\bp,N;\bal)=1$ determine the polynomials $P^{K}_{\bnu}(\bx;\bp,N;\bal)$ uniquely. Note that in \cite{I5}, $\al_0=0$, $w_j$ corresponds to $-1/\beta_j$, and $\nu_j$ to $n_j$. Using this correspondence, the eigenvalue of the $i$-th Gaudin operator can be computed by taking the coefficient of $\xi_i$ in \cite[Equation~(4.3c)]{I5} which coincides with the eigenvalue in \eqref{6.21}, so the polynomials defined by \eqref{6.19} there coincide with the polynomials defined by \eqref{6.23} above. We also outline an inductive proof of formula \eqref{6.19} similar to the proof of \eqref{4.14} in \thref{th4.4}, which reveals a new relation -- see equation~\eqref{6.26} below. First, note that for $w=w_k$ we can rewrite the operator $\fFK(w_k)$ as
\begin{equation}\label{6.24}
\fFK(w_k) =\fFK_1+ \fFK_2(w_k),
\end{equation}
where
\begin{equation*}
\fFK_1 = E_{N}^{-1} -\frac{1}{N}\sum_{s=1}^{d} x_s (\Id-E_{x_s}^{-1}) E_{N}^{-1} \quad\text{ and }\quad \fFK_2(w_k)=-\frac{1}{N} \sum_{s=1}^{d}\omega_{s,k} x_{s} E_{x_s}^{-1} E_{N}^{-1}.
\end{equation*}
Using the identity $a(\Id-E_{a}^{-1})\cdot(-a)_{k}=k(-a)_k$ one can show that if  $|\bnu|\leq N-1$ then the polynomials in \eqref{6.19} satisfy 
\begin{equation}
\fFK_1 P^{K}_{\bnu}(\bx;\bp,N;\bal) = P^{K}_{\bnu}(\bx;\bp,N;\bal).
\end{equation}
Combining the last equation with \eqref{6.24} we can deduce \eqref{6.19} from \eqref{6.23} by induction on $|\bnu|$ if we can show that for $|\bnu|\leq N-1$ the polynomials in \eqref{6.19} satisfy the relation
\begin{equation}\label{6.26}
P^{K}_{\bnu+\be_k}(\bx;\bp,N;\bal) -P^{K}_{\bnu}(\bx;\bp,N;\bal)  = \fFK_2(w_k) P^{K}_{\bnu}(\bx;\bp,N;\bal).
\end{equation}
Using the identity $(-a-1)_{k}-(-a)_{k}=-k(-a)_{k-1}$, we can write the left-hand side of \eqref{6.26} as
\begin{align*}
\mathrm{LHS}&=-\sum_{s=1}^{d}
\sum_{A=(a_{i,j})\in {\rM}_{d\times d}(N)}a_{s,k}\frac{(-\nu_k)_{\sum_{i=1}^{d}a_{i,k}-1}\prod_{j\neq k}(-\nu_j)_{\sum_{i=1}^{d}a_{i,j}}}{(-N)_{\sum_{i,j=1}^{d}a_{i,j}}}\\
&\qquad\qquad \times \prod_{i=1}^{d}(-x_i)_{\sum_{j=1}^{d}a_{i,j}} \prod_{i,j=1}^{d}\frac{\omega_{i,j}^{a_{i,j}}}{a_{i,j}!} \\
\intertext{and if $a_{s,k}>0$ we can replace the sum over $A$ with a sum over $B$ where  $A=B+E_{s,k}$}
&=-\sum_{s=1}^{d}
\sum_{B=(b_{i,j})\in {\rM}_{d\times d}(N-1)}\frac{\prod_{j=1}^{d}(-\nu_j)_{\sum_{i=1}^{d}b_{i,j}}}{(-N)_{\sum_{i,j=1}^{d}b_{i,j}+1}}\\
&\qquad\qquad \times (-x_s)_{\sum_{j=1}^{d}b_{s,j}+1} \, \prod_{i\neq s}(-x_i)_{\sum_{j=1}^{d}b_{i,j}}\,\omega_{s,k} \prod_{i,j=1}^{d}\frac{\omega_{i,j}^{b_{i,j}}}{b_{i,j}!}.
\end{align*}
Using the identity $aE_{a}^{-1}(-a)_{k}=-(-a)_{k+1}$ it follows that the right-hand side of \eqref{6.26} is equal to the last equation above, completing the proof of \eqref{6.26}.
\end{proof}

\begin{Remark}[Classical formula in dimension one]\label{re6.3}
Note that, unlike the Hahn case discussed in \reref{re4.6}, formula \eqref{6.19} easily reduces to the known formula for the Krawtchouk polynomials when $d=1$. Indeed, we have $\bx=x_1$, $\bnu=n$ and a short computation shows that $\omega_{1,1}=\frac{1}{p_{1}}$, hence formula \eqref{6.19} yields
\begin{align*}
P^{K}_{n}(x;\bp,N;\bal)= \sum_{a=0}^{n}\frac{(-n)_a(-x)_a}{(-N)_{a}\,a!}\,\frac{1}{p_1^a}=\tFo{-n}{-x}{-N}{\frac{1}{p_1}}.
\end{align*}
\end{Remark}

\begin{Remark}[Bispectrality]
More general multivariate Krawtchouk polynomials with respect to the multinomial distribution, depending on $d(d-1)/2$ free parameters, were introduced by Griffiths \cite{Gr1} in terms of a generating function. Mizukawa and Tanaka \cite{MT} gave an explicit formula for these polynomials using the Aomoto-Gelfand hypergeometric series  \cite{AK,Gel}. In \cite{I1} it was shown that the multivariate Krawtchouk polynomials defined by Griffiths are {\em bispectral}, i.e. they are common eigenfunctions of two families of commuting partial difference operators: one acting on the variables $\bx$ and the other one on the degree indices $\bnu$. The two families of operators can be connected by a bispectral involution $\mathfrak{b}$ defined on the variety which parametrizes the multivariate Krawtchouk polynomials. The results in \cite{I5} show that the subset which parametrizes the multinomial Gaudin models is invariant under the action of the bispectral involution $\mathfrak{b}$, which means that the polynomials $P^{K}_{\bnu}(\bx;\bp,N;\bal)$ are also common eigenfunctions of the operators for another Gaudin model acting on the degree indices $\bnu$. More precisely,  we can define dual parameters $\tilde{\boldsymbol{p}}=(\tilde{p}_{0},\dots,\tilde{p}_{d})$ satisfying  $\tilde{p}_{j}>0$, $|\tilde{\boldsymbol{p}}|=1$ and distinct real numbers $\beta_0,\dots,\beta_{d}$ such that 
$$P^{K}_{\bnu}(\bx;\bp,N;\bal)=P^{K}_{\bx}(\bnu;\tilde{\boldsymbol{p}},N;\boldsymbol{\beta}) \quad\text{ for }\quad \bnu,\bx\in V_N^d.$$
Therefore, $P^{K}_{\bnu}(\bx;\bp,N;\bal)$ will be eigenfunctions of the Gaudin operators in \eqref{6.17} acting on $\bnu$ instead of $\bx$, if we replace $\bp$ and $\bal$ with $\tilde{\boldsymbol{p}}$ and $\boldsymbol{\beta}$, respectively.
\end{Remark}

\begin{Remark}[Gaudin model for the negative multinomial distribution]\label{re6.5}
In this section, we focused on the multinomial distribution but clearly, the formula for the polynomials and the difference equations will hold for generic real or complex numbers $p_j$ and  $N$ such that $|\bp|=1$. More precisely, for generic $N$, $p_j$ and $\al_j$ we can define polynomials by formula \eqref{6.19}, where the sum on the right-hand side is over all matrices with entries in $\Nset_0$ (note that the series is terminating because $\nu_j\in\Nset_0$). Thus, we obtain polynomials defined for all $\bnu\in\Nset_0^d$ and equations~\eqref{6.20} and \eqref{6.26} hold for all $\bnu\in\Nset_0^d$. In particular, if we pick positive real numbers $s,c_1,\dots,c_d$ such that $|\mathbf{c}|=c_1+\cdots+c_d<1$ and if we set formally $N= -s$ and $p_j = - c_j /(1-|\mathbf{c}|)$ for $j=1,\dots,d$, we obtain polynomials defined for all $\bnu\in\Nset_0^d$ which are mutually orthogonal with respect to the negative multinomial distribution
$$\sM_{\mathbf{c},s}(\bx)=(1-|\mathbf{c}|)^{s}\,(s)_{|\bx|}\prod_{j=1}^{d}\frac{c_j^{x_j}}{x_j!}.$$ 
This leads to a Gaudin subclass of the multivariate Meixner polynomials introduced by Griffiths \cite{Gr2} which have a similar hypergeometric representation and bispectral properties established in \cite{I2}.
\end{Remark}

\begin{Remark}\label{re6.6}
The multinomial weight and the corresponding operators can be obtained from the multivariate Hahn weight and operators by a limit on the parameters as follows. For given parameters $\bp$ and $N$ of the multinomial distribution we set 
\begin{equation}\label{6.27}
\ka_i= t p_i, \qquad \text{ for }\qquad i=0,1,\dots,d.
\end{equation}
Then
$$\lim_{t\to\infty} \sH_{\bka,N} (\bx) =\sK_{\bp,N} (\bx) \quad\text{ and }\quad\lim_{t\to\infty}\frac{1}{t}\cL_{i,j}=\cLK_{i,j}.$$
For the Gaudin operators we have
\begin{equation}
\cGK_i(\bal)=\sum_{\begin{subarray}{c}j=0\\ j\neq i \end{subarray}}^{d}\frac{\cLK_{i,j}}{\al_{i}-\al_{j}}=\lim_{t\to\infty} \frac{1}{t} \cG_i(\bal), \qquad i=0,\dots, d.
\end{equation}
In particular, the operator for $i=0$ is
\begin{equation}\label{}
\cGK_0(\bal)=\lim_{t\to\infty} \frac{1}{t} \cG_0(\bal)=\sum_{j=1}^{d}\frac{1}{\al_{0}-\al_{j}}\left(p_{j} (N-|\bx|) (E_{x_j}-\Id)+p_0 x_j (E_{x_j}^{-1}-\Id)\right).
\end{equation}
If we set $\al_j=p_0-p_{j}$ for $j=1,\dots, d$ the operator $\cGK_0(\bal)$ reduces to the one studied in \cite{Sasaki} where equation \eqref{6.20} for $i=0$ was analyzed in the context of birth and death processes. The operator $\cG_0(\bal)$ in \eqref{4.8} can be regarded as a natural generalization corresponding to the multivariate Hahn distribution, while the operator $\cGD_0(\bal)$ in \eqref{5.7} is a continuous analog for the Dirichlet distribution.
\end{Remark}

\begin{Remark}
It is an interesting and challenging problem to find similar connections between multivariate orthogonal polynomials and {$\mathfrak{gl}_{m|n}$}  Gaudin models \cite{HMVY,LM}. Multivariate super Krawtchouk polynomials, related to the representation theory of the general linear Lie superalgebra, were defined and studied in \cite{IZ}.
\end{Remark}

\section*{Acknowledgments} I am grateful to Evgeny Mukhin for insightful discussions and valuable suggestions that helped refine an earlier version of this manuscript.

\end{document}